\documentclass[11pt]{article}
\usepackage{fullpage,verbatim}

\newcommand{\qed}{\Box}

\newtheorem{lemma}{Lemma}
\newtheorem{observation}[lemma]{Observation}

 \newcommand{\blackslug}{\hbox{\hskip 1pt \vrule width 4pt height 8pt
depth 1.5pt \hskip 1pt}}
 \newcommand{\QED}{\quad\blackslug\lower 8.5pt\null}
 \newcommand{\smQED}{\quad
  \vrule width1pt height6pt depth-1pt \kern -1pt
  \vrule width4pt height6pt depth-5pt \kern -4pt
  \vrule width4pt height2pt depth-1pt \kern -1pt
  \vrule width1pt height6pt depth-1pt }
 \newenvironment{proof}{{\bf Proof:}}{\QED}

\newcommand{\oldcomment}[1]{}
\newcommand{\junk}[1]{}

\renewcommand{\qed}{\hspace{1em}\hfill\rule{.5em}{.5em}}
\newcommand{\rubberskip}{\par\addvspace{\smallskipamount}}
\renewenvironment{proof}{\rubberskip\noindent{\bf Proof:}\ }{\qed
\rubberskip\noindent}

\newtheorem{freeaux}{}

\newenvironment{free}{\begin{freeaux}\hspace{-.9em}\rm}{\end{freeaux}}
\newcommand{\proofbox}{\rule{.5em}{.5em}}
\renewenvironment{proof}{\begin{free}{\bf Proof:}\ }{\ 
\hfill\proofbox\end{free}}

\newpage

\newcommand{\nestlab}{??}
\newcounter{algctr}\renewcommand{\thealgctr}{\Alph{algctr}}

\newlength{\lwidth}
\newlength{\lwidthi}
\newcounter{litemi}
\renewcommand{\thelitemi}{\nestlab.\arabic{litemi}}

\newlength{\lwidthii}
\newlength{\itemind}
\newlength{\itemindii}
\newcounter{litemii}
\renewcommand{\thelitemii}{\thelitemi.\arabic{litemii}}

\newlength{\lwidthiii}
\newlength{\itemindiii}
\newcounter{litemiii}
\renewcommand{\thelitemiii}{\thelitemii.\arabic{litemiii}}

\newlength{\lwidthiv}
\newlength{\itemindiv}
\newcounter{litemiv}
\renewcommand{\thelitemiv}{\thelitemiii.\arabic{litemiv}}

\newlength{\lwidthv}
\newlength{\itemindv}
\newcounter{litemv}
\renewcommand{\thelitemv}{\thelitemiv.\arabic{litemv}}

\usepackage{amssymb,amsmath}
\newcommand{\drop}[1]{}

\newcommand{\RSHIFT}{\gg}
\newcommand{\LSHIFT}{\ll}

\newcommand{\OR}{\vee}
\newcommand{\AND}{\wedge}
\newcommand{\XOR}{\oplus}
\newcommand{\NOT}{\neg}
\newcommand{\PRED}{\textnormal{predecessor}}
\newcommand{\SUCC}{\textnormal{successor}}
\newcommand{\RANK}{\textnormal{rank}}
\newcommand{\MATCH}{\textnormal{match}}
\newcommand{\MSB}{\textnormal{msb}}
\newcommand{\LSB}{\textnormal{lsb}}

\newcommand{\SELECT}{\,\textnormal{select}}
\newcommand\mco{\multicolumn{1}{|c|}{0}}
\newcommand\mcl{\multicolumn{1}{|c|}{1}}
\newcommand\clk{\ceiling{\lg k}}

\newif\ifcom
\comtrue

\newif\iflong
\longtrue

\newif\ifshort
\shortfalse

\newcommand{\ttl}{{\texttt{1}}}
\newcommand{\tto}{{\texttt{0}}}
\newcommand{\ttq}{{\texttt{?}}}

\newcommand{\lpiece}[3]{{{#1}}\langle {{#2}}\rangle_{{#3}}}
\newcommand{\lpieceq}[2]{{{#1}}\langle {{#2}}\rangle_{1}}

\newcommand{\xor}{\oplus}

\newcommand{\floor}[1]{\lfloor {#1} \rfloor}
\newcommand{\ceiling}[1]{\lceil {#1} \rceil}

\newcommand{\rank}{\mbox{\rm rank}}

\title{Dynamic Integer Sets with Optimal Rank, Select, and Predecessor Search
\footnote{Presented with different formatting in {\em Proceedings of the 55nd IEEE Symposium on Foundations of Computer Science (FOCS)}, 2014, pp. 166--175}} 

\author{\fbox{Mihai P\v{a}tra\c{s}cu\footnote{Passed away 2012.} }\and Mikkel Thorup\thanks{Research partly supported by an Advanced Grant from
    the Danish Council for Independent Research under the Sapere Aude
    research carrier programme.}\\
University of Copenhagen}
\date{\ \vspace{-1.5cm}}
\begin{document}
\maketitle
\begin{abstract}
We present a data structure representing a dynamic set $S$ of
$w$-bit integers on a $w$-bit word RAM. With $|S|=n$ and $w\geq\log n$ and
space $O(n)$, we support the following
standard operations in $O(\log n/\log w)$ time:
\begin{itemize}
\item insert$(x)$ sets $S=S\cup \{x\}$.
\item delete$(x)$ sets $S=S\setminus \{x\}$.
\item predecessor$(x)$ returns $\max\{y\in S\mid y< x\}$.
\item successor$(x)$ returns $\min\{y\in S\mid y\geq x\}$.
\item rank$(x)$ returns $\#\{y\in S\mid y< x\}$.
\item select$(i)$ returns $y\in S$ with $\rank(y)=i$, if any.
\end{itemize}
Our $O(\log n/\log w)$ bound is optimal for dynamic rank and select,
matching a lower bound of Fredman and Saks [STOC'89]. 
When the word length is large, our time bound is also optimal for
dynamic predecessor, matching a static lower bound of Beame and Fich
[STOC'99] whenever $\log n/\log w=O(\log w/\log\log w)$. 

Technically, the most interesting aspect of  our data structure is that it
supports all operations in constant time for sets of size $n=w^{O(1)}$.  This
resolves a main open problem of Ajtai, Komlos, and Fredman [FOCS'83].
Ajtai et al. presented such a data structure in Yao's abstract
cell-probe model with $w$-bit cells/words, but pointed out that the
functions used could not be implemented. As a partial solution to the
problem, Fredman and Willard [STOC'90] introduced a {\em fusion node\/} that
could handle queries in constant time, but used polynomial time on the
updates. Despite inefficient updates, this lead them to the first
sub-logarithmic bounds for sorting and dynamic predecessor search.
We call our data structure for small sets a {\em dynamic fusion node}
as it does both queries and updates in constant time.
\end{abstract}

\section{Introduction}
We consider the problem of representing a dynamic set $S$ of
integer keys so that we can efficiently support the
following standard operations:
\begin{itemize}
\item insert$(x)$ sets $S=S\cup \{x\}$.
\item delete$(x)$ sets $S=S\setminus \{x\}$.
\item member$(x)$ returns $[x\in S]$.
\item predecessor$(x)$ returns $\max\{y\in S\mid y< x\}$.
\item successor$(x)$ returns $\min\{y\in S\mid y\geq x\}$. We could also
have demanded $y>x$, but then the connection to rank and select below 
would be less elegant.
\item rank$(x)$ returns $\#\{y\in S\mid y< x\}$.
\item select$(i)$ returns $y\in S$ with $\rank(y)=i$, if any. 
Then $\PRED(x)=\SELECT(\RANK(x)-1)$ and $\SUCC(x)=\SELECT(\RANK(x))$.
\end{itemize}
Our main result is a deterministic linear space data structure supporting
all the above operations in $O(\log n/\log w)$ time where
$n=|S|$ is the set size and $w$ is the word length. This is
on the word RAM which models what can be implemented in a programming
language such as C~\cite{KR78} which has been used for fast portable 
code since 1978. Word operations take constant time.
The word size $w$, measured in bits,
is a unifying parameter of the model. All integers considered 
are assumed to fit in a word, and with $|S|=n$, we assume
$w\geq\log n$ so that we can at least index the elements in $S$. 
The random access memory of the word RAM implies that we can allocate
tables or arrays of words, accessing entries in constant time using indices that
may be computed from keys. This feature is used in many classic algorithms,
e.g., radix sort \cite{Com29} and hash tables \cite{Dum56}.

A unifying word size was not part of the original RAM model
\cite{CookR73}. However, real CPUs are tuned to work on words and
registers of a certain size $w$ (sometimes two sizes, e.g., 32 and 64 bit
registers).  Accessing smaller units, e.g., bits, is
more costly as they have to be extracted from words. On the other
hand, with unlimited word size, we can use word operations to encode a
massively parallel vector operations unless we make some artificial
restrictions on the available operations \cite{PS82}. The limited word
size is equivalent to a ``polynomial universe'' restriction where all
integers used are bounded by a polynomial in the problem size and sum
of the integers in the input. Restrictions of this type are present in
Yao's cell-probe model \cite{yao81tables} and in Kirkpatrick and
Reisch's integer sorting \cite{KR84}.  Fredman and Willard
\cite{fredman93fusion} made the word length explicit in their seminal
$O(n\log n/\log\log n)=o(n\log n)$ sorting algorithm.

Our $O(\log n/\log w)$ bound is optimal for dynamic rank and select,
matching a lower bound of Fredman and Saks
\cite{fredman89cellprobe}\footnote{For rank, see Theorem 3 and remark after Theorem 4 
in \cite{fredman89cellprobe}. The lower bound for select follows by another
simple reduction to be presented here.}. The lower bound is in Yao's cell-probe model \cite{yao81tables} 
where we only pay for probing cells/words in memory.
The lower bound is for the query time and holds
even with polylogarithmic update times. Conceivably, one
could get the same query time with constant time updates. Our optimality is for the maximal
operation time including both queries and updates.

When the word length is large, our time bound is also optimal for
dynamic predecessor, matching a cell-probe lower bound of Beame and Fich
\cite{beame02pred} whenever $\log n/\log w=O(\log w/\log\log w)$.
This lower bound is for the query time in the static case with any
polynomial space representation of the set $S$. 
Contrasting rank and select, predecessor
admits faster solutions when the word length is small, e.g., van Emde
Boas' \cite{vEB77pred} $O(\log w)$ time bound. The predecessor bounds
for smaller word lengths are essentially understood if we allow
randomization, and then our new bound completes
the picture for dynamic predecessor search. With key length $\ell\leq
w$, we get that the optimal expected maximal operation time for dynamic
predecessor (no rank or select) is 
\[
\Theta\left(\max\left\{1,\min \left\{
  \frac{\log n}{\log w},\ 
  \frac{\log \frac{\ell}{\log w}}{\log \left( \log \frac{\ell}{\log w} 
    \textrm{\large~/} \log\frac{\log n}{\log w} \right)},\ 
  \log \frac{\log (2^\ell-n)}{\log w}
 \right\}\right\}\right).\]
Our clean $O(\log n/\log w)$ bound for dynamic predecessor
search should be compared with Fredman and Willard's  \cite{fredman93fusion} bound
of $O(\log n/\log w + \sqrt{\log n})$\footnote{{\cite{fredman93fusion}} describe it as $O(\log n/\log b + \log b)$ for any
  $b\leq w^{1/6}$, from which they get $O(\log n/\log\log n)$ since
  $w\geq\log n$.}. For larger words, it was improved to $O(\log
n/\log w + \log\log n)$ by Andersson and Thorup
\cite{andersson07exptrees}, but the $\log\log n$ term still prevents
constant time when $\log n=O(\log w)$.

Technically, the most interesting aspect of our data structure is that
it operates in constant time for sets of size $n=w^{O(1)}$.  This
resolves a main open problem of Ajtai, Komlos, and Fredman
\cite{ajtai84hashing}.  Ajtai et al.~presented such a data structure
in Yao's abstract cell-probe model with $w$-bit cells/words, but
pointed out that the functions used could not be implemented. As a
partial solution to the problem, Fredman and Willard
\cite{fredman93fusion} introduced a {\em fusion node\/} that could
handle queries in constant time, but used polynomial time on the
updates. Despite inefficient updates, this lead them to the first
sub-logarithmic bounds for dynamic predecessor searching mentioned above.  We call
our data structure for small sets a {\em dynamic fusion node} as it
does both queries and updates in constant time. 

Fredman and Willard \cite{fredman94atomic} later introduced atomic
heaps that using tables of size $s$ can handle all operations in
constant time for sets of size $(\log s)^{O(1)}$.  Our focus
is on understanding the asymptotic impact of a large word length. We only use space linear in the number of stored keys, but our capacity for sets of
size $w^{O(1)}$ with constant operation time is the largest
possible regardless of the available space. Such a large set
size is new even in the simple case of a deterministic dynamic dictionary with
just membership and updates in constant time.

Like Fredman and Willard \cite{fredman93fusion,fredman94atomic}, we do
use multiplication. Thorup \cite{Tho03} has proved that this is
necessary, even if we just want to support membership in constant time
for sets of size $\Omega(\log n)$ using standard
instructions. However, as in \cite{andersson99fusion} for the classic
static fusion nodes, if we allow self-defined non-standard operations,
then we can implement our algorithms using AC$^0$ operations only. The
original cell-probe solution from \cite{ajtai84hashing} does not seem to
have an AC$^0$ implementation even if we allow non-standard operations.

\drop{
Orthogonal to our work, Cohen et al. \cite{CohenFHK13} have looked into
succinct implementations of static fusion nodes, minimizing the extra 
space per key,  but sacrificing the constant operation time. 
Here we are satisfied with linear space like in the x fusion nodes.
}

We emphasize that our new dynamic fusion nodes generally both simplify
and improve the application of fusion nodes in dynamic settings, most notably
in the original fusion trees \cite{fredman93fusion} which with the
original fusion nodes had to switch to a different technique towards
the bottom of the trees (c.f. Section \ref{sec:set-n}), yet never got optimal
results for large word-lengths. Fusion nodes are useful in many
different settings as illustrated by Willard \cite{willard92fusion}
for computational geometry.

\paragraph{Contents and techniques} Our new dynamic fusion node takes starting point in ideas and techniques from \cite{ajtai84hashing,fredman93fusion} bringing
them together in a tighter solution that for sets of size up to
$k=w^{1/4}$ supports all operations in constant time using only
standard instructions. In Section \ref{sec:prelim} we review the
previous techniques, emphasizing the parts and ideas that we will
reuse, but also pointing to their limitations. In Section
\ref{sec:dyn-fusion}, we describe the new dynamic fusion nodes,
constituting our core technical contribution. An important feature is
that we use matching with don't cares to code the search in a compressed
trie. In Section
\ref{sec:set-n}, we show how to use our dynamic fusion nodes in a
fusion tree for arbitrary set size $n$, supporting all operations in
time $O(\log n/\log k)=O(\log n/\log w)$. This is $O(1)$ for 
$n=w^{O(1)}$.  Because our fusion nodes are dynamic, we avoid the
dynamic binary search tree at the bottom of the original fusion trees
\cite{fredman93fusion}. Our fusion trees are augmented to handle rank
and select, using techniques from
\cite{dietz89sums,patrascu04connect}. In Section \ref{sec:lower} we
describe in more details how we get matching lower bounds from
\cite{fredman89cellprobe,beame02pred}. Finally, in Section
\ref{sec:ran-pred}, we complete the picture for randomized dynamic
predecessor search with smaller key lengths using the techniques from
\cite{patrascu06pred}.

\section{Notation and review of previous techniques}\label{sec:prelim}
We let $\lg$ denote $\log_2$ and $m=\{0,...,m-1\}$. Integers are represented in $w$-bit words with the least significant bit the to the right.
We shall use $\tto$ and $\ttl$ to denote the bits $0$ and $1$, contrasting
the numbers $0$ and $1$ that in words are padded with $w-1$ leading $\tto$s.
When working with $w$-bit integers, we will use the standard arithmetic
operations $+$, $-$, and $\times$, all working modulo
$2^w$. We will use the standard bit-wise Boolean operations: $\AND$ (and), $\OR$ (or), $\XOR$ (exclusive-or), and $\NOT$ (bit-negation). 
We also have left shift $\LSHIFT$ and right shift $\RSHIFT$.

As is standard in assignments, e.g., in C \cite{KR78}, if we assign an $\ell_0$-bit number $x$ to a
$\ell_1$-bit number $y$ and $\ell_0<\ell_1$, then $\ell_1-\ell_0$ leading
$\tto$s are added. If $\ell_0>\ell_1$, the $\ell_0-\ell_1$ leading
bits of $x$ are discarded.

\paragraph{Fields of words}
Often we view words as divided into fields of some length $f$. We
then use $\lpiece x i f$ to denote the $i$th field, starting from
the right with $\lpiece x 0 f$ the right most field. Thus
$x$ represents the integer $\sum_{i=0}^{w-1} 2^i\lpiece x i 1$.
Note that fields can easily be masked out using regular instructions, e.g., 
\[\lpiece x i f= (x\RSHIFT (i\times f))\AND ((1\LSHIFT f)-1).\]
A field assignment like $\lpiece x i f=y$ is implemented as
\[x=(x\AND \NOT m)\OR ((y\LSHIFT (i\times f))\AND m)\textnormal{ where }
m=((1\LSHIFT f)-1)\LSHIFT(i\times f)\]
Similarly, in constant time, we can mask out intervals of fields, using
\begin{align*}
\lpiece x {i..j} f&= (x\RSHIFT 
(i\times f))\AND ((1\LSHIFT ((j-i)\times f))-1),\\
\lpiece x {i..*} f&= (x\RSHIFT (i\times f)).
\end{align*}
For two-dimensional divisions of words into fields, we use
the notation
\[\lpiece x {i,j} {g\times f} =\lpiece x {i\times g+j} {f}.\]

\paragraph{Finding the most and least significant bits in constant time}
We have an operation $\MSB(x)$ that for an
integer $x$ computes the index of its most significant set bit. Fredman
and Willard \cite{fredman93fusion} showed how to implement this in
constant time using multiplication, but $\MSB$ can also be
implemented very efficiently by assigning $x$ to a floating point
number and extract the exponent. 
A theoretical advantage to using the
conversion to floating point numbers is that we avoid the 
universal constants depending on $w$ used in \cite{fredman93fusion}.

Using $\MSB$, we can also easily find
the least significant bit of $x$ as
$\LSB(x)=\MSB((x-1)\XOR x)$.

\paragraph{Small sets of size \boldmath{$k$}}
Our goal is to maintain a dynamic set of size $k=w^{\Theta(1)}$. 
Ajtai et al.~ \cite{ajtai84hashing} used
$k=w/\log w$ whereas Fredman and Willard \cite{fredman93fusion} used
$k=w^{1/6}$. We will ourselves use $k=w^{1/4}$. The exact value makes no theoretical difference, as our overall bound is $O(\log n/\log k)=O(\log n/\log w)$, which is $O(1)$ for $k=w^{\Omega(1)}$ and $n=w^{O(1)}$.

For simplicity, we assume below that $k$ is fixed, and we define various constants
based on $k$, e.g., $(\tto^{k-1}\ttl^k)^k$. However, using doubling, our
constants can all be computed in $O(\log k)$ time. We let $k$ be a power of two,
that we double as the sets grow larger. Then the cost of
computing the constants in the back ground is negligible.

\paragraph{Indexing}
We will store our key set $S$ in an unsorted array $KEY$ with
room for $k$ $w$-bit numbers. We will also maintain an array $INDEX$ of
$\ceiling{\lg k}$-bit indices so that $\lpiece {INDEX} i {\ceiling{\lg k}}$ is the index 
in $KEY$ of the key in $S$ of rank $i$ in the sorted order. An important point
here is that $INDEX$ with its $k\ceiling{\lg k}$ bits fits in a single word. Selection
is now trivially implemented as 
\[\SELECT(i)=KEY\left[\lpiece {INDEX} i \clk\right].\]
When looking for the predecessor of a query key $x$, we will
just look for its rank $i$ and return $\SELECT(i)$.

Consider the insertion of a new key $x$, supposing that we have found
its rank $i$. To support finding a free slot in $KEY$, we maintain
a bit map $bKEY$ over the used slots, and use $j=\MSB(bKEY)$ as the first free
slot. To insert $x$, we set $KEY[j]=x$ and $\lpiece {bKEY} {j} 1=0$.
To update $INDEX$, we set $\lpiece {INDEX} {i+1..k} {\ceiling{\lg k}}
=\lpiece {INDEX} {i..k-1} {\ceiling{\lg k}}$ and $\lpiece {INDEX} {i} 
{\ceiling{\lg k}}=j$. Deleting a key is just reversing the above.

\paragraph{Binary trie}
A binary trie \cite{fredkin60trie} for $w$-bit keys is a binary tree where 
the children of an internal node are labeled
$\tto$ and $\ttl$, and where all leaves are at depth $w$. A leaf
corresponds to the key represented by the bits along the path to the
root, the bit closest to the leaf being the least significant. Let
$T(S)$ be the trie of a key set $S$. We define the level $\ell(u)$ of an 
internal trie node $u$, to be the height of its children. 
To find the predecessor of a key
$x$ in $S$ we follow the path from the root of $T(S)$, matching the
bits of $x$.  When we get 
to a node $u$ in the search for $x$, we pick the child labeled $\lpiece x {\ell(u)} 1$.
If $x\not\in S$, then at some stage we get to a node $u$ with
a single child $v$ where $\lpiece x {\ell(u)} 1$ is opposite the bit of $v$. 
If $\lpiece x {\ell(u)} 1=1$, the predecessor of $x$ in $S$
the largest key below $v$.  Conversely, if $\lpiece x {\ell(u)} 1=0$, the successor of $x$ in $S$ is the smallest key below $v$.

\paragraph{Compressed binary trie}
In a compressed binary trie, or Patricia trie \cite{morrison68patricia},
we take the trie and short cut all paths of degree-1 nodes so
that we only have degree-2 branch nodes and leaves, hence at most $2k-1$
nodes.  This shortcutting does not change the level $\ell$ 
of the remaining nodes.

To search a key $x$ in the compressed trie $T^c(S)$, we
start at the root.  When we are at node $u$, as in the trie,
we go to the child labeled $\lpiece x {\ell(u)} 1$. Since all interior nodes have two
children, this search always finds a unique leaf corresponding to some key $y\in S$,
and we call $y$ the match of $x$.

The standard observation about compressed tries is that if $x$ matches
$y\in S$, then no key $z\in S$ can have a longer common prefix with $x$
than $y$. To see this, first note that if $x\in S$, then the search
must end in $x$, so we may assume that $x$ matches some $y\neq x$.
Let $j$
be the most significant bit where $x$ and $y$ differ, that is,
$j=\MSB(x\xor y)$.  If $z$ had a longer common prefix with $x$ then it
would branch from $y$ at level $j$, and then the search for $x$ should
have followed the branch towards $z$ instead of the branch to $y$. The
same observation implies that there was no branch node on level $j$
above the leaf of $y$. Let $v$ be the first node below level $j$ on the search
path to $y$. As for the regular trie we conclude that if $\lpiece x
{\ell(u)} 1=1$, the predecessor of $x$ in $S$ the largest key below $v$.
Conversely, if $\lpiece x {\ell(u)} 1=0$, the successor of $x$ in $S$
is the smallest key below $v$.

To insert the above $x$, we would just insert a branch node $u$ on level $j$
above $y$. The parent of $u$ is the previous parent of
$v$. A new leaf corresponding to $x$ is the $\lpiece x j 1$-child of
$u$ while $v$ is the other child.

\paragraph{Ajtai et al.'s cell probe solution}
We can now easily describe the cell probe solution of Ajtai et
al. \cite{ajtai84hashing}. We are going to reuse the parts 
that can be implemented on a word RAM, that is, the parts that only
need standard word operations.

As described previously, we assume that our set $S$ is stored in 
an unordered array $KEY$ of $k$ words plus a single word $INDEX$ such that
$KEY[\lpiece {INDEX} i \clk]$ is the key of rank $i$ in  $S$.

We represent the compressed trie in an array $T^c$ with $k-1$ entries
representing the internal nodes, the first entry being the root. 
Each entry needs its
level in  $[w]$ plus an index in $[k]$ to each child, or nil if the child is a leaf. 
We thus need
$O(\log w)$ bits per node entry, or $O(k\log w)=o(w)$ bits for the
whole array $T^c$ describing the compressed trie. The ordering of 
children induces an ordering of the leaves, and 
leaf $i$, corresponds to key $KEY[\lpiece {INDEX} i \clk]$.

In the cell-probe model, we are free to define arbitrary word operation
involving a constant number of words, 
including new operations on words representing compressed tries. We define
a new word operation $TrieSearch(x,T^c)$ that given a key $x$ and a correct
representation of a compressed trie $T^c$, returns the index $i$ of the
leaf matching $x$ in $T^c$. We then look up $y=KEY[\lpiece {INDEX} i \clk]$, and
compute $j=\MSB(x\XOR y)$. A second new word operation $TriePred(x,T^c,j)$ returns the rank $r$
of $x$ in $S$, assuming that $x$ branches of
from $T^c$ at level $j$.

To insert a key $x$ in the above data structure, we first compute the rank $r$ of $x$ in $S$ and then we
insert $x$ in $KEY$ and $INDEX$, as described previously.
A third new word operation $TrieInsert(T^c,x,j,r)$ changes $T^c$,
inserting a new leaf at rank $r$, adding the appropriate branch
node on level $j$ with child $\lpiece x j 1$ being the new leaf.

Summing up, in the cell probe model, it is easy to support all our
operations on a small dynamic set of integers, exploiting that a
compressed trie can be coded in a word, hence that we can navigate
a compressed trie using specially defined word operations.  Ajtai et al. 
\cite{ajtai84hashing} write ``{\em Open Question.} We conclude by stating the following problem. Our query
algorithm is not very realistic, as searching through a trie requires more than
constant time in actuality''.

\paragraph{Fredman and Willard's static fusion nodes with compressed keys}
Fredman and Willard \cite{fredman93fusion} addressed the above
problem, getting the searching in to constant time using only standard
instructions, but with no efficient updates. Their so-called {\em
  fusion node\/} supports constant time predecessor searches among up
to $w^{1/6}$ keys. However, with $k$ keys, it takes $O(k^4)$ time for them to
construct a fusion node.

Fredman and Willard proved the following lemma that we shall use repeatedly:
\begin{lemma}\label{lem:rank} Let $mb\leq w$. 
If we are given a $b$-bit number $x$ and a word $A$ with  $m$ $b$-bit numbers
 stored in sorted order, that is, $\lpiece A
0 b<\lpiece A 1 b<\cdots<\lpiece A {m-1} b$, then in constant time, 
we can find the rank of $x$ 
in $A$, denoted
$\RANK(x,A)$. 
\end{lemma}
The basic idea behind the lemma is to first use
multiplication to create $x^m$ consisting of $m$ copies of $x$; then
use a subtraction to code a coordinate-wise comparison $\lpiece A
i b<x$ for every $i$, and finally use $\MSB$ to find the largest such  $i$.

Next, they note that for the ordering of the keys in $S$, it suffices
to consider bits in positions $j$ such that for some $y,z\in S$,
$j=\MSB(y\xor z)$.  These are exactly the positions where we have
branch node somewhere in the trie.  Let $c_0<..<c_{\ell-1}$ be these
significant positions. Then $\ell<k$.

After computing some variables based on $S$, they show that in
constant time, they can compress a key $x$ so as to (essentially) only
contain the significant positions (plus some irrelevant $\tto$s).  Their
compressed key $\hat x$ is of length $b\leq k^3$. Let $\hat S$ denote a
sorted array with the compressed versions of the keys from $S$. With
$k^4\leq w$, the array fits in a single word, and they can therefore
compute $i=\rank(\hat x,\hat S)$ in constant time. Then $\lpiece {\hat S} i b$ is
the predecessor of $\hat x$ in $\hat S$, and $\lpiece {\hat S} {i+1} b$ is the
successor. One of these, say $\hat y=\lpiece {\hat S} i b$, 
has the longest common prefix with
$\hat x$, and they argue that then the original key $y=KEY\left[\lpiece {INDEX} i \clk\right]$ must also have
the longest common prefix with $x$ in $S$.

Next, they compute 
$j=\MSB(x\XOR y)$ and see where it fits between the significant bits,
computing $h=\RANK(j,(c_0,...,c_\ell))$, that is, $c_h<j\leq c_{h+1}$.
Referring to the compressed trie representation, $i$ and $h$ tells us
exactly in which shortcut edge $(u,v)$, the key $x$ branches out.
The predecessor of $x$ is the largest key below $y$ if $\lpiece x j 1=1$;
otherwise, the successor is the smallest key below $y$. The important
point here is that they can create a table $RANK$ that for each value
of $i\in [k]$, $h\in [k]$, and $\lpiece x j 1\in [2]$, returns the rank $r\in [k]$
of $x$ in $S$. From the rank we get the predecessor 
$KEY\left[\lpiece {INDEX} r \clk\right]$, all in constant time.

Fredman and Willard \cite{fredman93fusion} construct the above
fusion node in $O(k^4)$ time. While queries take
constant time, updates are no more efficient than computing the fusion
node from scratch.  
\drop{
Fredman
and Willard \cite{fredman94atomic} later presented their atomic heaps
using preprocessed tables of size $s$ handling all operations for
set sizes up to $(\log s)^{1/5}$.  The basic points is that
they can tabulate any operation involving two
$(\log s)/2$ bit integers. In this paper, we study the asymptotic
impact of a large word length $w$ which is not bounded by
the logarithm of the available space. The space we use is
only linear in the size of the set represented. }

\section{Dynamic fusion nodes}\label{sec:dyn-fusion}
We are now going to present our dynamic fusion nodes that for given
$k\leq w^{1/4}$ maintains a dynamic set $S$ of up to $k$ $w$-bit
integers, supporting both updates and queries in constant time using
only standard operations on words, thus
solving the open problem of Ajtai et al. \cite{ajtai84hashing}. 

As before, we
assume that our set $S$ is stored in an unordered array $KEY$ of $k$
words plus a single word $INDEX$ such that $\SELECT(i)=KEY[\lpiece
  {INDEX} i \clk]$ is the key of rank $i$ in $S$. As described
in the introduction, we can easily compute predecessor and
successor using rank and select.

\subsection{Matching with don't cares}
One of the problems making Fredman and Willard's fusion nodes dynamic
is that when a key is inserted, we may get a new significant position
$j$.  Then, for every key $y\in S$, we have to insert $\lpiece y j 1$ into
the compressed key $\hat y$, and there seems to be no easy way to
support this change in constant time per update.  In contrast, with a compressed
trie, when a new key is inserted, we only have to include a single new
branch node. However, as stated by Ajtai et al.
\cite{ajtai84hashing}, searching the compressed trie takes more than
constant time.

Here, we will introduce don't cares in the compressed keys so as to
make a perfect simulation of the compressed trie search in constant
time, and yet support updates in constant time.  Introducing don't
cares may seem rather surprising in that they normally make problems
harder, not easier, but in our context of small sets, it turns out
that using don't cares is just the right thing to do.  The basic idea
is illustrated in Figure \ref{fig:representation}.

\begin{figure}
\begin{center}
\hspace{-5cm}{\tt \begin{tabular*}{3ex}{*{8}{c}|c}
7&6&5&4&3&2&1&0 \\\cline{1-9}
1&1&1&1&1&0&1&0&4\\
1&1&0&1&1&0&1&0&3\\
1&1&0&1&0&0&1&0&2\\
1&0&0&1&0&0&1&0&1\\
1&0&0&1&0&0&0&1&0\\
\end{tabular*}}\\[2ex]
An array with four $8$-bit keys
with bit positions on top and ranks on the right.\medskip

\hspace{-5cm}{\tt
\begin{tabular*}{3ex}{*{8}{c}|c}
7&6&5&4&3&2&1&0 \\\cline{1-9}
 & &\mcl& & & & &&4 \\\cline{3-3}\cline{5-5}
 & &\mco& &\mcl& & & &3\\\cline{2-5}
 &\mcl& & &\mco& & &&2 \\\cline{1-2}\cline{5-5}\cline{7-7}
 &\mco& & & & &\mcl&&1 \\\cline{2-7}
 & & & & & &\mco& &0\\\cline{7-7} 
\end{tabular*}}\\[2ex]
Compressed trie as used by Ajtai et al.~\cite{ajtai84hashing}
\ \bigskip

\hspace{-4cm}{\tt \begin{tabular*}{3ex}{*{4}{c}|c}
6&5&3&1 \\\cline{1-5}
1&1&1&1&4\\
1&0&1&1&3\\
1&0&0&1&2\\
0&0&0&1&1\\
0&0&0&0&0\\[-2ex]
\rule{3ex}{0ex}&\rule{3ex}{0ex}&\rule{3ex}{0ex}&\rule{3ex}{0ex}&\rule{3ex}{0ex}
\end{tabular*}}\\[2ex]
Compressed keys as used by Fredman and Willard \cite{fredman93fusion}.
\ \bigskip

\hspace{-4cm}{\tt \begin{tabular*}{3ex}{*{4}{c}|c}
6&5&3&1 \\\cline{1-5}
1&1&?&?&4\\
1&0&1&?&3\\
1&0&0&?&2\\
0&?&?&1&1\\
0&?&?&0&0\\[-2ex]
\rule{3ex}{0ex}&\rule{3ex}{0ex}&\rule{3ex}{0ex}&\rule{3ex}{0ex}&\rule{3ex}{0ex}
\end{tabular*}}\\[2ex]
Compressed keys but with don't cares (\texttt?) in positions that are not
used for branching in the (compressed) trie.

\ \bigskip

\hspace{-4cm}{\tt \begin{tabular*}{3ex}{*{4}{c}|c}
6&5&3&1 \\\cline{1-5}
1&1&$x_3$&$x_1$&4\\
1&0&1&$x_1$&3\\
1&0&0&$x_1$&2\\
0&$x_5$&$x_3$&1&1\\
0&$x_5$&$x_3$&0&0\\[-2ex]
\rule{3ex}{0ex}&\rule{3ex}{0ex}&\rule{3ex}{0ex}&\rule{3ex}{0ex}&\rule{3ex}{0ex}
\end{tabular*}}\\[2ex]
When searching a key $x$, we replace don't cares 
at position $j$, with the $j$th bit $x_j$ of $x$.
\end{center}
\caption{Representations of a given key set.}\label{fig:representation}
\end{figure}

Initially, we ignore the key compression, and focus only on our
simulation of the search in a compressed trie.  For a key $y\in S$, we
will only care about bits $\lpiece y j 1$ such that $j=\MSB(y \XOR z)$
for some other $z\in S$. An equivalent definition is that if the trie has a
branch node $u$ at some level $j$, then we care about position $j$ for
all keys $y$ descending from $u$. We will now create a key $y^{\ttq}$ from $y$
using characters from $\{\tto,\ttl,\ttq\}$. Even though these characters 
are not just bits, we let $\lpieceq {y^{\ttq}} j$ denote the $j$th 
character of $y$. We
set $\lpieceq {y^{\ttq}} j=\ttq$ if we do not care about position $j$ in $y$;
otherwise $\lpieceq {y^{\ttq}} j=\lpiece y j 1$. 

A query key $x$ {\em matches\/} $y^{\ttq}$ if and only if $\lpiece x j
1=\lpieceq {y^{\ttq}} j$ for all $j$ such that $\lpieceq {y^{\ttq}} j\neq
\ttq$. Then $x$ matches $y^{\ttq}$ if and only if we get to $y$
when searching $x$ in the compressed trie over $S$.  Since such a search
in a compressed trie is always successful, we know that every $x$ has
a match $y^{\ttq}$, $y\in S$.
\begin{observation}\label{obs:trie-search} 
Let $y^{\ttq}$ be the match of  $x$, and suppose $y\neq x$. 
Set $j=\MSB(x \XOR y)$.  If $x<y$, then $x\AND (\ttl^{w-j}\tto^j)$ matches
$z^{\ttq}$ where $z$ is the successor of $x$. If $x>y$, then $x\OR
(\tto^{w-j}\ttl^j)$ matches $z^{\ttq}$ where $z$ is the predecessor of $x$.
\end{observation}
\begin{proof}
Since $x$ matches $y^{\ttq}$, we know that
$\lpiece y j 1=\ttq$. Let $(u,v)$ be the shortcut in the compressed trie
bypassing position $j$. If $x<y$, the successor of $x$ is the smallest
key below $v$, which is exactly what we find when we search $x\AND (\ttl^{w-j}\tto^j)$ since we always pick the $\tto$-child starting from $v$.
Likewise if  $x>y$, the predecessor of $x$ is the largest
key below $y$, which is exactly what we find when we search $x\OR 
(\tto^{w-j}\ttl^j)$ since we always pick the $\ttl$-child starting from $v$.
\end{proof}
Observation \ref{obs:trie-search} implies that two matchings
suffice to find the predecessor or successor of a key $x$. The
second matching eliminates the need of the
table RANK in the original static fusion nodes described above.  This
two step approach is similar to the ``blind'' compressed trie search for
integers by Grossi et al.~\cite{grossi09succinct}. However,
Grossi et al.~do not use our matching with don't cares. Their data structure is static without updates, and they use 
tables of size $2^{\Omega(w)}$ like in Fredman and Willard's 
atomic heaps \cite{fredman94atomic}. 

Next we observe that the key compression of fusion nodes does not
affect the matching of keys with don't cares, for if we for
some key care for position $j$, then by definition this is a significant
position. Conversely this means that an insignificant position is
a don't care in all keys, so skipping insignificant positions does not
affect the matching, that is, $x$ matches $y^{\ttq}$ if and only
if $\hat x$ matches $\hat y^{\ttq}$ where $\hat x$ and $\hat y^{\ttq}$ are
the compressed version of $x$ and $y^{\ttq}$.

In Section \ref{sec:key-compress}, we will show how to maintain a 
perfect compression which keeps the significant positions only.
Thus, if $\ell\leq k$ is the number of significant positions, then
$\hat x$ and $\hat y^{\ttq}$ will both of length $\ell$ (recall that Fredman
and Willard \cite{fredman93fusion} had $k^3$ bits in their compressed keys,
and they did not show how to update the
compression when new keys got inserted). For consistent
formatting, we view compressed keys as having $k$ bits with $k-\ell$ leading $\tto$s.

To check if $\hat x$ matches $\hat y^{\ttq}$, we will create
$\hat y^{\hat x}$ such that $\lpiece {\hat y^{\hat x}} h 1=\lpiece {\hat y} h 1$ if $\lpiece {\hat y^{{\ttq}}} h 1\neq\ttq$ while $\lpiece {\hat y^{\hat x}} h 1=\lpiece {\hat x} h 1$ if $\lpiece {\hat y^{{\ttq}}} h 1=\ttq$.
Then $\hat x$ matches $\hat y^{\ttq}$ if and only if 
$\hat x=\hat y^{\hat x}$.

Note that if $y,z\in S$ and $x$ is any $w$-bit key, then $y<z\iff 
\hat y^{\hat x}<\hat z^{\hat x}$. This is because the bits from
$x$ do not introduce any new branchings between the keys $y,z\in S$.

We will create $\hat y^{\hat x}$ for all $y\in S$ simultaneously in
constant time. To this end, we maintain two $k\times k$ bit matrices
$BRANCH$ and $FREE$ that together represent the stored compressed keys
with don't cares. Let $y$ be the key of rank $i$ in $S$, and let $c_h$
be a significant position.  If $\lpieceq {\hat y^{\ttq}} h\neq \ttq$,
then $\lpiece {BRANCH} {i,h} {k\times 1}= \lpiece {\hat y} h 1$
and $\lpiece {FREE} {i,h} {k\times 1}=\tto$. If $\lpieceq {\hat
  y^{\ttq}} h={\ttq}$, then $\lpiece {BRANCH} {i,h} {k\times 1}=\tto$ and
$\lpiece {FREE} {i,h} {k\times 1}=\ttl$. For a compressed key $\hat
x$, let $\hat x^k$ denote $x$ repeated $k$ times, obtained 
multiplying $x$ by the constant $(\tto^{k-1}\ttl)^k$.
Then $BRANCH\OR(\hat x^k\AND FREE)$ is the desired array of the $\hat y^{\hat x}$ for $y\in S$, and
we define
\[\MATCH(x)=\RANK(\hat x,BRANCH\OR(\hat x^k\AND FREE)).\]
Here $\RANK$ is implemented as in Lemma \ref{lem:rank}, and 
then $\MATCH(x)$ gives the index of the $y^{\ttq}$ matching $x$. By
Observation \ref{obs:trie-search}, we can
therefore compute the rank of $x$ in $S$ as follows.
\paragraph{$\RANK(x)$}
\begin{itemize}
\item $i=\MATCH(x)$.
\item $y=KEY\left[\lpiece {INDEX} i \clk\right]$.
\item if $x=y$ return $i$.
\item $j=\MSB(x\XOR y)$.
\item $i_0=\MATCH(x\AND (\ttl^{w-j}\tto^j))$.
\item $i_1=\MATCH(x\OR (\tto^{w-j}\ttl^j))$.
\item if $x<y$, return $i_0-1$.
\item if $x>y$, return $i_1$.
\end{itemize}
\paragraph{Inserting a key}
We will only insert a key $x$ after we have applied $\RANK(x)$ as above,
verifying that $x\not\in S$. This also means that we have
computed $j$, $i_0$ and $i_1$, where $j$ denotes
the position where $x$ should branch off as a leaf child. The keys
below the other child are exactly those with indices $i_0,...,i_1$.

Before we continue, we have to consider if $j$ is a new significant
position, say with rank $h$ among the current significant positions.
If so, we have to update the compression function to include $j$, as we
shall do it in Section \ref{sec:key-compress}. The first 
effect on the current compressed key arrays is to insert a column $h$ of
don't cares. This is a new column of $\tto$s in the $BRANCH$ array, 
and a column of $\ttl$s in the $FREE$ array. To achieve this effect,
we first have to shift all columns $\geq h$ one to the left. To do this,
we need some simple masks. The $k\times k$ bit matrix $M_h$ with column $h$ set
is computed as
\[M_h=(\tto^{k-1}\ttl)^k\LSHIFT h\textnormal,\]
and the $k\times k$ bit matrix $M_{i:j}$ with columns $i,...,j$ set are 
computed as
\[M_{i:j}=M_{j+1}-M_i.\]
This calculation does not work for $M_{0:k-1}$ which instead
is trivially computed as $M_{0:k-1}=\ttl^{k^2}$. For $X=BRANCH,FREE$, we want to shift all columns $\geq h$ one to the left.
This is done by 
\[X=(X\AND M_{0:h-1})\OR((X\AND M_{h:k-1})\LSHIFT 1).\]
To fix the new column $h$, we set 
\begin{align*}
FREE&=FREE\OR M_h\\
BRANCH&=BRANCH\AND\NOT M_h.
\end{align*}
We have now made sure that column $h$ corresponding to position $j=c_h$
is included as a significant position in $BRANCH$ and $FREE$.
Now, for all keys with indices $i_0,...,i_1$, we want to set
the bits in column $h$ to $\lpiece y j 1$. To do this,
we compute the $k\times k$ bit mask $M^{i_0:i_1}$ with rows
$i_0,...,i_1$ set
\[M^{i_0:i_1}=(1\LSHIFT((i_1+1)\times k))-(1\LSHIFT(i_0\times k))\]
and then we mask out column $h$ within these rows by
\[M^{i_0:i_1}_h=M^{i_0:i_1}\AND M_h.\]
We now set 
\begin{align*}
FREE&=FREE\AND\NOT M^{i_0:i_1}_h\\
BRANCH&=BRANCH\OR (M^{i_0:i_1}_h\times \lpiece y j 1)
\end{align*}
We now need to insert the row corresponding to $\hat x^{\ttq}$ with 
$r=\RANK(x)$. 
Making room for row $r$
in $X=BRANCH,FREE$, is done by
\[X=(X\AND M^{0:r-1})\OR((X\AND M^{r:k-1})\LSHIFT k).\]
We also need to create the new row corresponding to $\hat x^{\ttq}$.
Let $\hat y^{\ttq}$ be the compressed key with don't cares that
$x$ matched. Then
\begin{align*}
\lpieceq {\hat x^{\ttq}} {0..h-1}&= \ttq^h\\
\lpieceq {\hat x^{\ttq}} {h} &= \lpiece x {j} 1\\
\lpieceq {\hat x^{\ttq}} {h+1..k-1} &= \lpieceq {\hat y^{\ttq}} {h+1..k-1} 
\end{align*}
With $i$ the index of $y$ after the above insertion of row $r$ for $x$,
we thus set
\begin{align*}
\lpiece {BRANCH} {r,0..h-1} {k\times 1}&= \tto^h\\
\lpiece {FREE} {r,0..h-1} {k\times 1}&= \ttl^h\\
\lpiece {BRANCH} {r,h} {k\times 1}&= \lpiece x {j} 1\\
\lpiece {FREE} {r,h} {k\times 1}&= \tto\\
\lpiece {BRANCH} {r,h+1..k-1} {k\times 1}&= \lpiece {BRANCH} {i,h+1..k-1} {k\times 1}\\
\lpiece {FREE} {r,h+1..k-1} {k\times 1}&= \lpiece {FREE} {i,h+1..k-1} {k\times 1}
\end{align*}
This completes the update of $BRANCH$ and $FREE$. To delete a key, we
just have to invert the above process.
We still have to explain our new perfect compression.

\subsection{Perfect dynamic key compression}\label{sec:key-compress}
As Fredman and Willard \cite{fredman93fusion}, we have a set
of at most $k$ significant bit positions $c_0,...,c_{\ell-1}\in [w]$,
$c_0<c_1<\cdots < c_{\ell-1}$. For a given key $x$, we want
to construct, in constant time, an $\ell$-bit compressed key $\hat x$ such that
for $h\in[\ell]$, $\lpiece {\hat x} h 1=\lpiece x {c_h} 1$. We will also
support that significant positions can be added or removed in constant
time, so as to be added or excluded in future compressions.

Fredman and Willard had a more loose compression, allowing the
compressed key to have some irrelevant zeros between the significant bits,
and the compressed keys to have length $O(k^3)$. 
Here, we will only keep the significant bits. 

Our map will involve three multiplications. The first multiplication
pack the significant bits in a segment of length $O(w/k^2)$. The last two
multiplications will reorder them and place them consecutively.

\subsection{Packing}
We pack the significant bits using the approach of Tarjan and Yao 
for storing a sparse table \cite{tarjan79sparse}. We have a $w$-bit word $B$ where the significant bits are 
set. For a given key $x$, we will operate on $x\AND B$.

For some parameter $b$, we divide words into blocks of
$b$ bits. For our construction below, we can pick any $b$ such that
$k^2\leq b\leq w/k^2$. We are going to pack the significant bits into a
segment of $2b$ bits, so a smaller $b$ would seem to be an
advantage. However, if we later want a uniform algorithm that works
when $k$ is not a constant, then it is important that the number of
blocks is $k^{O(1)}$. We therefore pick the maximal $b=w/k^2$.

We are looking for $w/b$ shift values $s_i<k^2/2$ such that
no two shifted blocks $\lpiece B i {k^2}\LSHIFT s_i$ have a set bit in 
the same position. The packed key will be the $2b$ bit segment
\[\mu(x)=\sum_{i=0}^{w/b}(\lpiece {(x\AND B)} i b\LSHIFT s_i).\]
Since we are never adding set bits in the same position, the  set
bits will appear directly in the sum. The shifts $s_i$ are found
greedily for $i=0,..,w/b-1$. We want to pick $s_i$ such that
$\lpiece B i b\LSHIFT s_i$ does not intersect with 
$\sum_{j=0}^{i-1}(\lpiece B j b\LSHIFT s_j)$. A set bit $p$ in $\lpiece B i b$
collides with set bit $q$ in $\sum_{j=0}^{i-1}(\lpiece B j b\LSHIFT s_j)$
only if $s_i=p-q$. All together,
we have at most $k$ set bits, so there at most $k^2/4$ shift values $s_i$ 
that we cannot use. Hence there is always a good shift value $s_i<k^2/2$.

So compute the above sum in constant time, we will maintain a word $BSHIFT$
that specifies the shifts of all the blocks. For $i=0,...,w/b-1$, we
will have 
\[\lpiece {BSHIFT} {w/b-1-i} b=1\LSHIFT s_i.\]
Essentially we wish to compute the packing via the product $(x\AND
B)\times BSHIFT$, but then we get interference between blocks. To
avoid that we use a mask $E_b=(\tto^b\ttl^b)^{w/(2b)}$ to pick out
every other block, and then compute:
\begin{align*}
\mu(x)=&\left((x\AND B\AND E_b)\times (BSHIFT\AND E_b))\right.\\
&\left.+(((x\AND B)\RSHIFT b)\AND E_b)\times ((BSHIFT\RSHIFT b)\AND E_b)\right)\RSHIFT (w-2b).
\end{align*}
\paragraph{Resetting a block}
If new significant bit $j$ is included in $B$, we may have to change the
shift $s_i$ of the block $i$ containing bit $j$. With $b$ a power of two,
$i=j\RSHIFT \lg b$. 

The update is quite simple. First we remove block $i$ from the
packing by setting 
\[\lpiece {BSHIFT} {w/b-1-i} b=0.\]
Now $\mu(B)$ describes the packing of all bits outside block $i$.
To find a new shift value $s_i$, we use a constant
$S_0$ with 
$bk^2$ bits, where for $s=0,..,k^2/2-1$, 
\[\lpiece {S_0} s {2b}=1\LSHIFT s.\]
Computing the product $S_0\times \lpiece B i b$, we get
that 
\[\lpiece {(S_0\times \lpiece B i b)} s {2b}=\lpiece B i b \LSHIFT s\]
We can use $s$ as a shift if and only if
\[\mu(B)\AND\lpiece {(S_0\times \lpiece B i b)} s {2b}=0.\]
This value is always bounded by $(1\LSHIFT 3b/2)$,
so 
\begin{align*}
&\mu(B)\AND\lpiece {(S_0\times \lpiece B i b)} s {2b}=0\\
&\iff (1\LSHIFT 3b/2)\AND ((1\LSHIFT 3b/2)-(\mu(B)\AND\lpiece {(S_0\times \lpiece B i b)} s {2b}))\neq 0.
\end{align*}
We can thus compute the smallest good shift value as
\[s_i=\LSB((1\LSHIFT 3b/2)^{k^2/2}\AND ((1\LSHIFT 3b/2)^{k^2/2}-
(\mu(B)^{k^2/2}\AND (S_0\times \lpiece B i b)))\RSHIFT \lg (2b).\]
All that remains is to set 
\[\lpiece {BSHIFT} i b=1\LSHIFT s_i.\]
Then $\mu(x)$ gives the desired updated packing.

\subsection{Final reordering}
We now have all significant bits packed in $\mu(x)$ which has $2b\leq
2w/k^2$ bits, with significant bit $h$ in some position $\mu_h$. 
We want to compute the compressed key $\hat x$ with all significant bits in place, that is, $\lpiece {\hat x} h 1 = \lpiece {\hat x} {c_h} 1$.

The first step is to create a word $LSHIFT$ with $k$ fields of $4b$ bit fields
where $\lpiece {LSHIFT} h {4b}=1\LSHIFT (2b-\mu_h)$. 
Then $(LSHIFT\times\mu(x))\RSHIFT 2b$ is a word where the $h$th field has the
$h$th significant bit in its first position. We mask out
all other bits using $((LSHIFT\times\mu(x)) \RSHIFT 2b)\AND (\tto^{4b-1}\ttl)^{k}$.

Finally, to collect the significant bits, we multiply with the $k\times
4b$ bit constant  $S_1$ where for $h\in[k]$, $\lpiece {S_1} {k-1-h}
{4b}=\ttl\LSHIFT h$. Now the compressed key can be computed as
\[\hat x=\left(S_1\times \left(((LSHIFT\times\mu(x)) \RSHIFT 2b)\AND (\tto^{4b-1}\ttl)^{k}\right)\right)\RSHIFT((k-1)\times 4b).\]

\paragraph{Updates}
Finally we need to describe how to maintain $LSHIFT$ in constant time per update.
To this end, we will also maintain a $k\times \lg w$ bit word $C$ with the indices of the 
significant bits, that is, $\lpiece C h {\lg w}=c_h$.

We can view updates as partitioned in two parts: one is to change the shift of
block $i$, and the other is to insert or remove a significant bit. First
we consider the changed shift of block $i$ from $s_i$ to $s_i'$.  If block $i$ has no significant
bits, that is, if $\lpiece B i b=0$, there is nothing to be done.
Otherwise, the index of the first
significant bit in block $i$ is computed as $h_0=\RANK(i\times b,C)$, and the
index of the last one is computed as $h_1=\RANK((i+1)\times b,C)-1$.

In the packing, for $h=h_0,..,h_1$, significant bit $h$ will move
to $\mu_h'=\mu_h+s_i'-s_i$. The numbers $\mu_h$
are not stored explicitly, but only as shifts in the sense
that $\lpiece {LSHIFT} h {4b}=\ttl\LSHIFT (2b-\mu_i)$. The
change in shift is the same for all $h\in [h_0,h_1]$, so
all we have to do is to set
\[\lpiece {LSHIFT} {h_0..h_1} {4b}=\lpiece {LSHIFT} {h_0..h_1} {4b}\LSHIFT (s_b-s'_b)\textnormal,\]
interpreting ``$\LSHIFT (s_b-s'_b)$'' as ``$\RSHIFT (s'_b-s_b)$'' if $s_b<s_b'$.

The other operation we have to handle is when we get a new significant bit $h$
is introduced at position $c_h\in[w]$, that is, $\lpiece B {c_h} 1=\ttl$. 
First we have to insert $c_h$ in $C$,
setting 
\[\lpiece C {h+1..*} {\lg w}=\lpiece C {h..*} {\lg w}\textnormal{ 
and }\lpiece C {h} {\lg w}=c_h.\] 
Significant bit $h$ resides in block $i=c_h\RSHIFT \lg b$, and the
packing will place it in position $\mu_i=c_h-i\times b+s_i$.
Thus for the final update of $LSHIFT$, we set
\[\lpiece {LSHIFT} {h+1..*} {4b}=
\lpiece {LSHIFT} {h..*} {4b}\textnormal{ 
and }\lpiece {LSHIFT}  {h} {4b}=1\LSHIFT (2b-\mu_i).\] 
Removal of significant bit $h$ is accomplished by the simple statements:
\[\lpiece C {h..*} {\lg w}=\lpiece C {h+1..*} {\lg w}\textnormal{ 
and }\lpiece {LSHIFT} {h..*} {4b}=
\lpiece {LSHIFT} {h+1..*} {4b}.\]
This completes our description of perfect compression.

\section{Dynamic fusion trees with operation time $O(\log n/\log w)$}\label{sec:set-n}
We will now show how to maintain a dynamic set $S$ of $n$ $w$-bit integers,
supporting all our operations in $O(\log n/\log w)$ time. Our
fusion nodes of size $k$ is used in a fusion tree which is a search tree
of degree $\Theta(k)$. Because our fusion nodes are
dynamic, we avoid the dynamic binary search tree at the bottom of the
original dynamic fusion trees \cite{fredman93fusion}. Our fusion trees are
augmented to handle rank and select using the approach of 
Dietz \cite{dietz89sums} for maintaining
order in a linked list. Dietz, however,  was only dealing with a list and
no keys, and he used tabulation that could only handle nodes of
size $O(\log n)$. To exploit a potentially larger word length $w$, we will
employ a simplifying trick of P{\v a}tra{\c s}cu and Demaine \cite[\S 8]{patrascu04connect} for partial sums over small weights.
For simplicity, we are going to represent the set $S\cup\{-\infty\}$ where
$-\infty$ is a key smaller than any real key.

Our dynamic fusion node supports all operations in constant time
for sets of size up to $k=w^{1/4}$. We will create a fusion search tree
with degrees between $k/16$ and $k$. The keys are all in the leaves, and
a node on height $i$ will have between $(k/4)^i/4$ and $(k/4)^i$ descending 
leaves. It standard
that constant time for updates to a node imply that we
can maintain the balance for a tree in time proportional to the 
height when leaves are inserted or deleted. When siblings are merged or split, we build new node structures in
the background, adding children one at the time (see, e.g., \cite[\S
  3]{andersson07exptrees} for details).
The total height is 
$O(\log n/\log w)$, which will bound our query and update time.

First, assume that we only want to support predecessor search. For
each child $v$ of a node $u$, we want to know the smallest key descending
from $v$. The set $S_u$ of these smallest keys is maintained
by our fusion node. This includes $KEY_u$ and $INDEX_h$ as described 
in Section \ref{sec:prelim}. The $i$th
key is found as $KEY_u[\lpiece {INDEX_u} i \clk]$. We will have another
array $CHILD_u$ that gives us pointers to the children of $u$, where $CHILD_u$ uses
the same indexing as $KEY_u$. The $i$th child is thus pointed to
by $CHILD_u[\lpiece {INDEX_u} i \clk]$. 

When we come to a fusion node $u$ with a key $x$, we get in constant time
the index $i=\RANK_u(x)$ of the child $v$ such that $x$ belongs under, that is,
the largest $i$ such that $KEY[\lpiece {INDEX} i \clk]< x$.
Following $CHILD_u[\lpiece {INDEX} i \clk]$ down, we eventually get
to the leaf with the predecessor of $x$ in $S$.

Now, if we also want to know the rank of $x$, then for each node
$u$, and child index $i$ of $u$, we want to know the number $N_u[i]$ of keys
descending from the children indexed $<i$. Adding up these
numbers from all nodes on the search path gives us the rank of $x$ (as a technical detail, the count will also count $-\infty$ for one unless the search ends
$-\infty$).

The issue is how we can maintain the numbers $N[i]$; for when we add a
key below the child indexed $i$, we want to add one to $N[j]$ for all
of $j>i$. As in \cite[\S 8]{patrascu04connect} we will maintain
$N_u[i]$ as the sum of two numbers.  For each node $u$, we will have a
regular array $M_u$ of $k$ numbers, each one word, plus an array $D_u$
of $k$ numbers between $-k$ and $k$, each using $\clk+1$ bits. We will
always have $N_u[i]=M_u[i]+\lpiece {D_u} i {\clk+1}$. As described in
\cite[\S 8]{patrascu04connect}, it easy in constant time to add 1 to
$\lpiece {D_u} j {\clk+1}$ for all $j>i$, simply by setting
\[\lpiece {D_u} {i+1..*} {\clk+1}=\lpiece {D_u} {i+1..*} {\clk+1} +(\tto^{\clk}\ttl)^{k-i}.\]
Finally, consider the issue of selecting the key with rank $r$. This
is done recursively, staring from the root. When
we are at a node $u$, we look for the largest $i$ such that $N_u[i]
< r$.
Then we set $r=r-N_u[i]$, and continue the search from child $i$ of $u$. 
To find this child index $i$, we adopt a nice trick of Dietz
\cite{dietz89sums} which works when $u$ is of height $h\geq 2$. We use
$M_u$ as a good approximation for $N_u$ which does not change too often.
We maintain an array $Q_u$ that for each $j\in [k]$ stores the largest index
$i'=Q[j]$ such that $M_u[i']< j k^{h-1}/4^h$. Then the right index $i$ is at most
$5$ different from $i'=Q_u[r 4^h/k^{h-1}]$, so we can check the $10$ relevant indices
$i$ exhaustively. Maintaining $Q_u$ is easy since we in round-robin fashion only
change one $M_u[j]$ at the time, affecting at most a constant number
of $Q_u[j]$.
For nodes of height $1$, we do not need to do anything; for the children
are the leaf keys, so select $r$ just returns the child of index $r$.

This completes the description of our dynamic fusion tree supporting rank, 
select, and predecessor search, plus updates, all in time $O(\log n/\log w)$.

\section{Lower bounds}\label{sec:lower}
First we note that dynamic cell-probe lower bounds hold regardless of the available
space. The simple point is that if we start with an empty set and insert $n$ elements in $O(\log n)$ amortized time per element, then, using perfect
hashing, we could make an $O(n\log n)$ space data structure, that
in constant time could tell the contents of all cells written based on
the integers inserted. Any preprocessing done before the integers were inserted
would be a universal constant, which is ``known'' to the cell probe
querier. The cost of the next dynamic query can thus
be matched by a static data structure using $O(n \log n)$ space.

We claimed that Fredman and Saks \cite{fredman89cellprobe}
provides an $\Omega(\log n/\log w)$ lower bound for dynamic rank and select.
Both follow by reduction from dynamic prefix parity where
we have an array $B$ with $n$ bits, starting, say, with all zeros. An update
flips bit $i$, and query asks what is the xor of the first $j$ bits. By
\cite[Theorem 3]{fredman89cellprobe}, for $\log n\leq w$, the 
amortized operation time is $\Omega(\log n/\log w)$.

\paragraph{Rank} Fredman and Saks mention that a lower bound for rank follows by an
easy reduction: simply consider a dynamic set $S\subseteq [n]$,
where $i\in S$ iff $B[i]=1$. Thus we add $i$ to $S$ when we set
$B[i]$, and delete $i$ from $S$ when we unset $B[i]$. The prefix
parity of $B[0,..,j]$ is the parity of the rank of $j$ in $S$.

\paragraph{Select} Our reduction to selection is bit more convoluted. We assume
that $n$ is divisible by 2. For dynamic select, we
consider a dynamic set $S$ of integers from $[n(n+1)]$. which
we initialize as the set $S_0$ consisting of the integers
\[S_0=\{i(n+1)+j+1\mid i,j\in [n]\}\]
We let bit $B[i]$ in prefix parity correspond to 
integer $i(n+1)\in S$. When we want to ask for prefix parity
of $B[0,..,j]$, we set $y=$select$((j+1)n-1)$.
If there are $r$ set bits in $B[0,..,j]$, then $y=(j+1)(n+1)-1-r$, so 
we can compute $r$ from $y$.

To understand that the above is a legal reduction, note that
putting in $S_0$ which is a fixed constant does not help. While it is
true that the data structure for dynamic select can do work when $S_0$
is being inserted, this is worth nothing in the cell probe model, for
$S_0$ is a fixed constant, so in the cell-probe, if this was useful,
it could just be simulated by the data structure for prefix parity, at
no cost.

\paragraph{Predecessor}
The query time for dynamic predecessor with logarithmic update time
cannot be better than that of static predecessor using polynomial space,
for which Beame and Fich  \cite[Theorem 3.6 plus reductions in Sections 3.1-2]{beame02pred} have proved
a lower bond of $\Omega(\log n/\log w)$  if
$\log n/\log w=O(\log (w/\log n)/\log\log (w/\log n))$. But
the condition is equivalent to $\log n/\log w=O(\log w/\log\log w)$.
To see the equivalence, note that if $w\geq\log^2 n$ then
$\log (w/\log n)=\Theta(\log w)$ and  if $w\leq\log^2 n$, both forms of the
conditions are satisfied. Sen and Vankatesh \cite{sen08roundelim} have proved that the lower bound holds even if we allow randomization.

\section{Optimal randomized dynamic predecessor}\label{sec:ran-pred}
From \cite{patrascu06pred,patrascu07randpred} we know that with key
length $\ell$ and word length $w$, the optimal (randomized) query time for static predecessor search using $\tilde O(n)$ space for $n$ keys is within a constant factor of
\begin{align*}
\max\{1,\min \left\{ \begin{array}{l}
  \frac{\log n}{\log w} \\[1.5ex]
  \frac{\mbox{$\log \frac{\ell}{\log w}$}}{\log \left( \log \frac{\ell}{\log w} 
    \textrm{\large~/} \log\frac{\log n}{\log w} \right)}\\[3.5ex]
  \log \frac{\ell-\lg n}{\log w}
\end{array} \right.
\end{align*}
This is then also the best bound we can hope for dynamically with
logarithmic update times, for we can use the dynamic data structure to
build a static data structure, simply by inserting the $n$ keys.
However, in the introduction, we actually actually replaced the last
branch $\log \frac{\ell-\lg n}{\log w}$ with the higher bound $\log
\frac{\log (2^\ell-n)}{\log w}$. To get this improvement from the static lower
bound, we do as follows.  First note that we can assume that
$\log n\geq \ell/2$, for otherwise the bounds are the same. We
now set $n'=\sqrt{2^\ell-n}<n$ and
$\ell'=\floor{\lg (2^\ell-n)}$.  For a set $S'\subseteq[2^{\ell'}]$ of
size $n'$ using $\tilde O(n')$ space, the static lower bound states
that the query time is $\Omega(\log \frac{\ell'-\lg n'}{\log
  w})=\Omega(\log \frac{\log (2^\ell-n)}{\log w})$. To create such a
data structure dynamically, first, in a preliminary step, for
$i=1,...,n-n'$, we insert the key $2^\ell-i$. We have not yet inserted
any keys from $[2^{\ell'}]$. Now we insert the keys from $S'$. It is
only the $\tilde O(|S'|)$ cells written after we start inserting $S'$
that carry any information about $S'$. Everything written before is
viewed as a universal constant not counted in the cell-probe model.
The newly written cells and their addresses,
form a near-linear space representation of $S'$, to which the
static cell-probe lower bound of $\Omega(\log \frac{\ell'-\lg n'}{\log
  w})=\Omega(\log \frac{\log (2^\ell-n)}{\log w})$ applies.

The main result of this paper is that the top $O(\log n/\log w)$
branch is possible. Essentially this was the missing dynamic bound for
predecessor search if we allow randomization. 

Going carefully through
Section 5 of the full version of \cite{patrascu06pred}, we see that
the lower two branches are fairly easy to implement dynamically if we
use randomized hash tables with constant expected update and query
time (or alternatively, with constant query time and constant
amortized expected update time \cite{dietzfel94dph}). The last branch
is van Emde Boas' data structure \cite{vEB77pred,Mehlhorn90boas}
with some simple tuning. This includes the space saving Y-fast tries
of Willard \cite{willard83pred}, but using our dynamic fusion nodes in
the bottom trees so that they are searched and updated in constant
time. This immediately gives a bound of $O(\log \frac{\ell}{\log
  w})$. The improvement to $O(\log \frac{\log (2^\ell-n)}{\log w})$ is
obtained by observing that for each key stored, both neighbors are at
distance at most $\Delta=2^\ell-n$.
This, in turns, implies that the van Emde
data structure has a recursive call at the prefix of length
$\ell-\ell'$ where $\ell'=2^{\lceil{\lg\lg \Delta\rceil}}$. All such
prefixes will be stored, and for a query key we can look up the
$\ell-\ell'$ bit prefix and start the recursion from there, and then
we complete the search in $O(\log \frac{\ell'}{\log w})=O(\log
\frac{\log (2^\ell-n)}{\log w})$ time. 

To implement the middle part, we follow the small space static
construction from \cite[Full version, Section 5.5.2]{patrascu06pred}.
For certain parameters $q$ and $h$, \cite{patrascu06pred} uses a special data structure 
\cite[Full version, Lemma 20]{patrascu06pred} (which is based on \cite{beame02pred}) using $O(q^{2h})$ space, which is also the construction time. 
They use $q=n^{1/(4h)}$ for
a $O(\sqrt n)$ space and construction time. 
The special data structure only has to
change when a subproblem of size $\Theta(n/q)=\Theta(n^{1-1/(4h)})$
changes in size by a constant factor, which means that
the $O(\sqrt n)$ construction time gets amortized
as sub-constant time per key update. This is similar to
the calculation behind exponential search trees \cite{andersson07exptrees}
and we can use the same de-amortization. Other than the special table,
each node needs a regular hash table like van Emde Boas' data structure, as discussed above.

For a concrete parameter choice, we apply our dynamic
fusion node at the bottom to get $O(w)$ free space per key. Then,
as in \cite{patrascu06pred},
\[h=\frac{\lg \frac {\ell}a}
{\lg\frac{\lg n}{a}}\left/\lg \frac{\lg \frac {\ell}a}
{\lg\frac{\lg n}{a}}\right.
\textnormal{ where }a=\lg w.\]
This yields the desired recursion depth of 
\[O\left(\frac{\lg \frac {\ell}{\lg w}}{\lg h}+h\lg\frac{\lg n}{\lg w}\right)=
O\left(\frac{\mbox{$\log \frac{\ell}{\log w}$}}{\log \left( \log \frac{\ell}{\log w} 
    \textrm{\large~/} \log\frac{\log n}{\log w} \right)}\right).
\]

\paragraph{Acknowledgment} Would like to thank
Djamal Belazzougui for pointing out a bug in one of the bounds stated for 
predecessor search in an earlier version of this paper.

\newcommand{\etalchar}[1]{$^{#1}$}

\end{document}